\documentclass[review]{elsarticle}
\usepackage[most]{tcolorbox}

\usepackage{adjustbox}
\usepackage{natbib}
\usepackage{caption}
\usepackage{lineno}
\usepackage{mathtools}
\usepackage{graphicx}
\usepackage{lineno}
\usepackage{amssymb,mathtools}
\usepackage{expl3}
\usepackage{float}
\usepackage[caption = false]{subfig}
\usepackage{longtable}
\usepackage{hyperref}
\usepackage{caption}
\usepackage{amssymb}
\usepackage{amsthm}
\usepackage{amsmath}
\usepackage{lipsum}
\usepackage{multicol}
\usepackage{graphicx}
\usepackage{longtable}
\usepackage{multirow}
\usepackage{comment}
\usepackage{textcomp}
\usepackage{lineno,hyperref}
\modulolinenumbers[5]

\theoremstyle{definition}
\newtheorem{defn}{Definition}
\usepackage{algorithm}
\usepackage{multirow}
\usepackage[margin=0.7in]{geometry}
\usepackage{amsfonts}
\usepackage{amsmath}
\usepackage{amsthm}
\usepackage{amssymb}
\usepackage{epsfig}
\usepackage{mathrsfs}

\newcommand{\Z}{\mathbb{Z}}
\newcommand{\N}{\mathbb{N}}
\newtheorem{thm}{Theorem}[section]

\newtheorem{exm}{Example}[section]
\newtheorem{rem}{Remark}[section]

\newtheorem{cor}{Corollary}[section]

\makeatletter
\def\ps@pprintTitle{%
   \let\@oddhead\@empty
   \let\@evenhead\@empty
   \def\@oddfoot{\reset@font\hfil\thepage\hfil}
   \let\@evenfoot\@oddfoot
}
\makeatother
\begin{document}
\begin{frontmatter}

\title{Relationship of Two Discrete Dynamical Models: One-dimensional Cellular Automata and Integral Value Transformations}
\author[mymainaddress1,mymainaddress2]{Sreeya Ghosh}
\ead{sreeya135@gmail.com}
\author[mymainaddress3]{Sudhakar Sahoo}
\ead{sudhakar.sahoo@gmail.com}
\author[mymainaddress4]{Sk. Sarif Hassan\corref{cor1}}
\ead{sarimif@gmail.com}
\author[mymainaddress5,mymainaddress2]{Jayanta Kumar Das} 
\ead{dasjayantakumar89@gmail.com}
\author[mymainaddress1]{Pabitra Pal Choudhury}
\ead{pabitra@isical.ac.in}
\cortext[cor1]{Corresponding author}
\address[mymainaddress1]{Department of Applied Mathematics, University of Calcutta, Kolkata-700009, India}
\address[mymainaddress2]{Applied Statistics Unit, Indian Statistical Institute, Kolkata-108 , India}

\address[mymainaddress3]{Institute of Mathematics and Applications, Bhubaneshwar-751029, India}
\address[mymainaddress4]{Dept. of Mathematics, Pingla Thana Mahavidyalaya, Paschim Medinipur-721140, India}
\address[mymainaddress5]{School of Medicine, Johns Hopkins University, MD-21210, USA}

\begin{abstract}
Cellular Automaton(CA) and an Integral Value Transformation(IVT) are two well established mathematical models which evolve in discrete time steps. Theoretically, studies on CA suggest that CA is capable of producing a great variety of evolution patterns. However computation of non-linear CA or higher dimensional CA maybe complex, whereas IVTs can be manipulated easily.\\
The main purpose of this paper is to study the link between a transition function of a one-dimensional CA and IVTs. Mathematically, we have also established the algebraic structures of a set of transition functions of a one-dimensional CA as well as that of a set of IVTs using binary operations. Also DNA sequence evolution has been modeled using IVTs.
\end{abstract}

\end{frontmatter}

\section{Introduction}
Cellular Automaton(pl. cellular automata, abbrev. CA)  is a discrete model which has applications in computer science, mathematics, physics, complexity science, theoretical biology and microstructure modeling. This model was introduced by J.von Neumann and S.Ulam in 1940 for designing self replicating systems \cite{von1966theory,schiff2011cellular,ulam1962some}.

A CA consists of a finite/countably infinite number of finite-state semi-automata \cite{ghosh2016evo,ghosh2016finite} known as \lq cells\rq~arranged in an ordered $n$-dimensional grid. Each cell receives input from the neighbouring cells and changes according to the transition function. The transitions at each of the cells together induces a change of the grid pattern. The simplest CA is a CA where the grid is a one-dimensional line. Stephen Wolfram's work in the 1980s contributed to a systematic study of one-dimensional CA, providing the first qualitative classification of their behaviour \cite{wolfram2002new}. CA has been studied for solving many interesting problems on utilizing mathematical bases such as polynomial, matrix algebra, Boolean derivative etc.\cite{choudhury2009investigation,das1992characterization}. Further, dynamic behaviour of CA can also be studied using various mathematical tools \cite{wuensche1997attractor,wuensche2003discrete,xu2009dynamical,edwards2019class}, that help to understand the crucial properties and modeling of various classes of discrete dynamical system \cite{edwards2019class}.

An Integral Value Transformation (abbrev. IVT), a class of discrete dynamical system, were first introduced during 2009-10 (\cite{hassan2010collatz}). A IVT  of $k$-dimension is a function defined using $p$-adic numbers over $\N_0^k$ where $\N_0=\N \cup \{0\}$, $p \in \N$. The IVTs have been studied in different viewpoints including the understanding of the evolution of integer sequences and behavioral patterns of integers in discrete time points \cite{hassan2011integral,das2017two}. 

In this paper, Wolfram code of an elementary CA and global transition function of a one-dimensional CA has been represented through IVTs. The algebraic structure of a set of transition functions of a one-dimensional CA as well as the algebraic structure of a set of IVTs under some binary operations have been studied. In the last section, some specific applications have been discussed in which we have seen certain novelty lies in IVTs over one-dimensional CA.
 
\section{Mathematical preliminaries}
\begin{defn}
	Let  \textbf{$Q$} be a finite set of memory elements also called the \textbf{state set}.\\ A \textbf{global configuration} is a mapping from the group of integers ${\Z}$ to the set $Q$ given by $C: {\Z} \rightarrow Q$. The set $Q^{\Z}$ is the set of all global configurations where $Q^{\Z} = \{C|C : {\Z} \rightarrow Q\}$.
\end{defn}

\begin{defn}
	A mapping $\tau : Q^{\Z} \rightarrow Q^{\Z}$  is called a  \textbf{global transition function}.\\A \textbf{CA}(denoted by $\mathcal{C}_\tau^Q$)(reported in \cite{ghosh2017some,kari2005theory}) is a triplet $(Q, Q^{\Z}, \tau)$, where,
	\begin{itemize}
		\item $Q$ is the finite state set
		\item $Q^{\Z}$ is the set of all configurations
		\item $\tau$ is the global transition function
	\end{itemize}
	\begin{defn}
		The set $\mathscr{Q}^{\Z}= \{\tau | \tau : Q^{\Z} \rightarrow Q^{\Z}\}$ is the set of all possible global transition functions of  CA having state set $Q$.
	\end{defn}
	A mapping $\tau$ is invertible if $\forall C_i, C_j \in Q^\Z$, $\exists \tau^{-1}$ such that $$\tau(C_i)=C_j \Leftrightarrow \tau^{-1}(C_j)=C_i$$  
\end{defn}

\begin{defn}
	For $i\in {\Z}, r\in {\N}$,  let $S_i = \{i-r,...,i-1,i,i+1,...,i+r \}\subseteq {\Z}$. $S_i$ is the neighbourhood of the $i^{th}$ cell, $r$ is the radius of the neighbourhood of a cell. It follows that ${\Z} = \bigcup_iS_i$.\\A restriction from ${\Z}$ to $S_i$ induces a restriction of $C$ to $\overline{c_i}$ given by $\overline{c_i}:S_i \rightarrow Q$; where $\overline{c_i}$ may be called \textbf{local configuration} of the $i^{th}$ cell.\\
	The mapping $\mu_i : Q^{S_i} \rightarrow Q$ is known as a \textbf{local transition function} for the $i^{th}$ cell having radius $r$. Thus $ \forall i \in { \Z},\mu_i(\overline{c_i})\in Q$  and it follows that,
	$$ \tau(C)=\tau(...,c_{i-1},c_i,c_{i+1},...)=.....\mu_{i-1}(\overline{c_{i-1}}).\mu_i(\overline{c_i}).\mu_{i+1}(\overline{c_{i+1}}).......$$
\end{defn}
\begin{defn}
	The set $M= \{\mu_i | \mu_i : Q^{S_i} \rightarrow Q, i \in \Z\}$ is the set of all possible local transition functions of CA having state set $Q$.
\end{defn}
\begin{defn} If for a particular CA, $ |Q|=2$ so that we can write $Q=\{0, 1\}$, then the CA is said to be a \textbf{binary CA} or a \textbf{Boolean CA}. A Boolean CA having radius $1$ is known as an \textbf{Elementary CA(ECA)}. 
\end{defn}

\begin{defn}
	\textbf{Wolfram code} is a naming system often used for a one-dimensional CA introduced by Stephen Wolfram (see\cite{wolfram2002new}).\\For a one-dimensional CA with $Q$ states, the local rule $\mu_i$ for some $i^{th}$ cell, $i \in {\Z}$ of radius $r$(neighbourhood $2r+1$) can be specified by an $Q^{2r+1}$-bit sequence. The decimal equivalent form of this sequence is known as the Wolfram code. 
\end{defn}
Thus, the Wolfram code for a particular rule is a number in the range from $0$ to $Q^{Q^{2r+1}}-1$, converted from $Q$-ary to decimal notation.
\begin{exm}
	Let the local rule of an ECA for some $i^{th}$ cell, $i \in {\Z}$ be,
	$$\mu_i(\overline{c_i}) = \mu_i(c_{i-1},c_i,c_{i+1}) = (c_{i-1} \vee c_{i+1}) \wedge c_i$$ where $'\vee'$ stands for OR operation, $'\wedge'$ stands for AND operation, $c_j$ is the $j^{th}$ cell configuration for $j=i-1, i, i+1$.\\ Then we get the $2^3$-bit sequence as, $$\mu(111)\mu(110)\mu(101)\mu(100)\mu(011)\mu(010)\mu(001)\mu(000)~~=~~ 1~1~0~0~1~0~0~0$$
	The decimal equivalent number for $11001000$ is $200$ and so the Wolfram code is RULE \textbf{200}.
\end{exm}

\begin{defn}
	A $p$-adic $k-$dimensional \textbf{Integral Value Transformation(IVT)} denoted by $IVT_j^{p,k}$ for $p \in \N$, $k \in \N$ is a function of $p$-base numbers from $\N_0^k$ to $\N_0$ defined(in \cite{hassan2012onedmensional}) as $$IVT_j^{p,k}(n_1,n_2,...,n_k)=(f_j(a_0^{n_1},...,a_0^{n_k})f_j(a_1^{n_1},...,a_1^{n_k})...f_j(a_{l-1}^{n_1},...,a_{l-1}^{n_k}))_p=m$$ where $\N_0=\N \cup \{0\}$, $p \in \N$, $n_s=(a_0^{n_s}a_1^{n_s}...a_{l-1}^{n_s})_p$ for $s=1,2,...,k$.\\$f_j$ is a function from $\{0,1,...,p-1\}^k$ to $\{0,1,...,p-1\}$ for $j=0,1,...,p^{p^k}-1$,\\ $m$ is the decimal conversion of the $p$-base number.
\end{defn}
\begin{rem}
	In particular(discussed in \cite{choudhury2009theory,choudhury2011act,das2016multi,pal2012properties}), if $\forall i=0,1,...,l-1$, the function $f_j$ be defined as
	\begin{enumerate}
		\item $f_j(a_i^{n_1},...,a_i^{n_k})= \lfloor{(a_i^{n_1}+...+a_i^{n_k})/p}\rfloor$, then $IVT_j^{p,k}$ is known as a \textbf{Modified Carry Value Transformation(MCVT)}.\\Again, $MCVT$ with a $0$ padding at the right end is known as a \textbf{Carry Value Transformation(CVT)}.
		\item $f_j(a_i^{n_1},...,a_i^{n_k})= (a_i^{n_1}+...+a_i^{n_k})~mod~p$, then $IVT_j^{p,k}$ is known as an \textbf{Exclusive OR Transformation(XORT)}.
		\item $f_j(a_i^{n_1},...,a_i^{n_k})= max{(a_i^{n_1},...,a_i^{n_k})}$, then $IVT_j^{p,k}$ is known as an \textbf{Extreme Value Transformation(EVT)}.
	\end{enumerate}
\end{rem}

\section{Wolfram code of an ECA and IVT}
For an ECA, Wolfram code for a local rule $\mu$  is the decimal number $j \in \{0,1,2,...,255\}$ obtained from the  $8$-bit sequence $$\mu(111)\mu(110)\mu(101)\mu(100)\mu(011)\mu(010)\mu(001)\mu(000)=j$$ 
Therefore, Wolfram code for each local transition function of a $3$-neighbourhood Boolean CA can be represented by a $2$  base $3$-dimensional IVT.\\ 
The function $\mu$ in the above $8$-bit sequence can be represented by a function $f_j:\{0,1\}^3 \rightarrow \{0,1\}$ which gives $$(f_j(111)f_j(110)f_j(101)f_j(100)f_j(011)f_j(010)f_j(001)f_j(000))_2$$ $$=IVT_j^{2,3}(11110000_2,11001100_2,10101010_2)=IVT_j^{2,3}(240_{10},204_{10},170_{10})$$
Hence it follows that for a $3-$neighbourhood Boolean CA, any Wolfram code $j \in \{0,1,2,...255\}$ can be equivalently represented by $IVT_j^{2,3}(240_{10},204_{10},170_{10})$. 
\begin{exm}
	Wolfram code 200 can be equivalently represented as $$200_{10}=(11001000)_2=IVT_{200}^{2,3}(240_{10},204_{10},170_{10})$$
\end{exm}
However the following example shows that for $j \in \{0,1,2,...255\}$, any $IVT_j^{2,3}$ may not correspond to a Wolfram code.
\begin{exm}
	$$IVT_{j}^{2,3}(240_{10},204_{10},171_{10})=IVT_j^{2,3}(11110000_2,11001100_2,10101011_2)$$
	$$=(f_j(111)f_j(110)f_j(101)f_j(100)f_j(011)f_j(010)f_j(001)f_j(001))_2$$
	In this $8$-bit sequence, $f_j(000)$ is missing and so this cannot correspond to any Wolfram code.
\end{exm}
\subsection{Some Particular Wolfram codes and Particular IVTs}
We know that transformations such as MCVT, XORT, EVT are particular cases of IVTs. Again, any Wolfram code can be represented by an IVT. Therefore some particular Wolfram codes which can be represented by an MCVT, XORT or EVT are as follows.     
\begin{enumerate}
	\item Let the local rule of an ECA for some $i^{th}$ cell, $i \in {\Z}$ be,
	$$\mu_i(\overline{c_i}) = \mu_i(c_{i-1},c_i,c_{i+1}) = (c_{i-1} \wedge c_i) \vee((c_{i-1} \vee c_i) \wedge c_{i+1})$$ Then we get the $2^3$-bit sequence as $11101000$ and Wolfram code $232$.\\ Here, $\mu_i$ can be represented by function $f_j:\{0,1\}^3 \rightarrow \{0,1\}$ given by $$f_j(c_{i-1},c_i,c_{i+1})=(c_{i-1} \wedge c_i) \vee((c_{i-1} \vee c_i) \wedge c_{i+1})=\lfloor{(c_{i-1} + c_i + c_{i+1})/2}\rfloor$$ Thus, Wolfram code $232$ can be represented as $$(f_{232}(111)f_{232}(110)f_{232}(101)f_{232}(100)f_{232}(011)f_{232}(010)f_{232}(001)f_{232}(000))_2$$ $$=(\lfloor{3/2}\rfloor\lfloor{2/2}\rfloor\lfloor{2/2}\rfloor\lfloor{1/2}\rfloor\lfloor{2/2}\rfloor\lfloor{1/2}\rfloor\lfloor{1/2}\rfloor\lfloor{0/2}\rfloor)_2$$ $$=MCVT_{232}^{2,3}(11110000_2,11001100_2,10101010_2)=MCVT_{232}^{2,3}(240_{10},204_{10},170_{10})$$
	\item Let the local rule of an ECA for some $i^{th}$ cell, $i \in {\Z}$ be,
	$$\mu_i(\overline{c_i}) = \mu_i(c_{i-1},c_i,c_{i+1}) = (c_{i-1} \underline{\vee} c_i \underline{\vee} c_{i+1})$$ Then we get the $2^3$-bit sequence as $10010110$ and Wolfram code $150$.\\ Here, $\mu_i$ can be represented by function $f_j:\{0,1\}^3 \rightarrow \{0,1\}$ given by $$f_j(c_{i-1},c_i,c_{i+1})=(c_{i-1} \underline{\vee} c_i \underline{\vee} c_{i+1})=(c_{i-1} + c_i + c_{i+1})~mod~2$$ Thus, Wolfram code $150$ can be represented as $$(f_{150}(111)f_{150}(110)f_{150}(101)f_{150}(100)f_{150}(011)f_{150}(010)f_{150}(001)f_{150}(000))_2$$ $$=((3mod2)(2mod2)(2mod2)(1mod2)(2mod2)(1mod2)(1mod2)(0mod2))_2$$ $$=XORT_{150}^{2,3}(11110000_2,11001100_2,10101010_2)=XORT_{150}^{2,3}(240_{10},204_{10},170_{10})$$
	\item Let the local rule of an ECA for some $i^{th}$ cell, $i \in {\Z}$ be,
	$$\mu_i(\overline{c_i}) = \mu_i(c_{i-1},c_i,c_{i+1}) = (c_{i-1} \vee c_i \vee c_{i+1})$$ Then we get the $2^3$-bit sequence as $11111110$ and Wolfram code $254$.\\ Here, $\mu_i$ can be represented by function $f_j:\{0,1\}^3 \rightarrow \{0,1\}$ given by $$f_j(c_{i-1},c_i,c_{i+1})=(c_{i-1} \vee c_i \vee c_{i+1})=max\{c_{i-1}, c_i, c_{i+1}\}$$ Thus, Wolfram code $254$ can be represented as $$(f_{254}(111)f_{254}(110)f_{254}(101)f_{254}(100)f_{254}(011)f_{254}(010)f_{254}(001)f_{254}(000))_2$$ $$=(11111110)_2$$ $$=EVT_{254}^{2,3}(11110000_2,11001100_2,10101010_2)=EVT_{254}^{2,3}(240_{10},204_{10},170_{10})$$
\end{enumerate}

\section{Transition Function of a CA and IVT}
\begin{defn}
	The set $T^{p,k}=\{IVT_j^{p,k}|IVT_j^{p,k} : \N_0^k \rightarrow \N_0\}$ is the set of all $p$-base $k$-dimensional IVTs for $p \in \N, k \in \N$.
\end{defn}
\begin{defn}
	The $v$-fold cartesian product $\underbrace{T^{p,k} \times ... \times T^{p,k}}_{v~times}$ for $v \in \N$, denoted by $\mathscr{T}^v$ is given by $$\mathscr{T}^v=\{(IVT_{j_1}^{p,k},IVT_{j_2}^{p,k},...,IVT_{j_v}^{p,k})\}$$ where $IVT_{j_s}^{p,k}:\N_0^k \rightarrow \N_0 \in T^{p,k}$, for $s=1,2,...,v(\in \N)$
\end{defn}
\begin{defn}
	Let a restriction on $IVT_j^{p,k}$ from $\N_0^k$ to $Q^k=\{0,1,...,p-1\}^k$ for $p \in \N$ be denoted by $\underline{IVT}_j^{p,k}$. The set $T^{p,k}_{~|Q}=\{\underline{IVT}_j^{p,k}|\underline{IVT}_j^{p,k} : Q^k \rightarrow Q\}$ is the set of all $p$-base $k$-dimensional IVTs when $\N_0^k$ is restricted to the subset $Q^k$. Therefore for $(\eta_1,\eta_2,...,\eta_k) \in Q^k$ we get, $$\underline{IVT}_j^{p,k}(\eta_1,\eta_2,...,\eta_k)= f_j(\eta_1,\eta_2,...,\eta_k)_p= m \in Q$$ 
\end{defn}

\begin{defn}
	The set $\mathscr{T}^{v}_{~~|Q}=\{(\underline{IVT}_{j_1}^{p,k},\underline{IVT}_{j_2}^{p,k},...,\underline{IVT}_{j_v}^{p,k})\}$ is the $v$-fold cartesian product $\underbrace{T^{p,k}_{~|Q} \times ... \times T^{p,k}_{~|Q}}_{v~times}$ when $\N_0^k$ is restricted to subset $Q^k$ where $\underline{IVT}_{j_s}^{p,k}:Q^k \rightarrow Q \in T^{p,k}_{~|Q}$ for $s=1,2,...,v(\in \N)$
\end{defn}

Now, a local transition function of any $i^{th}$ cell of a one-dimensional CA having $p(\in \N)$ states and radius $r(\in \N)$ will be of the form $$\mu_i(\overline{c_i}) = \mu_i(c_{i-r},...,c_i,...,c_{i+r})=c_i^*$$ where $c_{i-r},...,c_i,...,c_{i+r}, c_i^* \in \{0,1,...,p-1\}$.\\
This  can be represented by some $\underline{IVT}_{j_i}^{p,k}$, where $p \in \N, k=2r+1$ and $j_i \in \{0,1,2,...,(p^{p^k}-1)\}$ is the underlying Wolfram code(which is in turn equal to $IVT_{j_i}^{p,k}(240_{10},204_{10},170_{10})$).\\ Thus $\forall i \in \Z$, $$\mu_i(\overline{c_i}) \equiv \underline{IVT}_{j_i}^{p,k}(\overline{c_i})$$
Now if $\tau$ be the global transition function of any $v(\in \N \geq k)$-celled CA, then for a global configuration $C=(c_1,c_2,...,c_v)$,  we get $$ \tau(C)=\tau(c_1,c_2,...,c_v)=\mu_{1}(\overline{c_1}).\mu_2(\overline{c_2})....\mu_{v}(\overline{c_v})$$
Therefore it follows that, $$ \tau(c_1,c_2,...,c_v) \equiv (\underline{IVT}_{j_1}^{p,k}(\overline{c_1}),\underline{IVT}_{j_2}^{p,k}(\overline{c_2}),...,\underline{IVT}_{j_v}^{p,k}(\overline{c_v}))$$ 
Hence for any $v$-celled  ECA it follows that, $$ \tau(C) \equiv (\underline{IVT}_{j_1}^{2,3}(\overline{c_1}),\underline{IVT}_{j_2}^{2,3}(\overline{c_2}),...,\underline{IVT}_{j_v}^{2,3}(\overline{c_v}))$$ where $j_1,j_2,...,j_v \in \{0,1,...,255\}$.\\
Now if $\tau$ be the global transition function of a countably infinite celled CA with cells having $k$ neighbourhood, then for any global configuration\\ $C=(...c_{i-1},c_i,c_{i+1},...)$,  we get 
$$ \tau(C)=.....\mu_{i-1}(\overline{c_{i-1}}).\mu_i(\overline{c_i}).\mu_{i+1}(\overline{c_{i+1}}).......$$
Therefore it follows that, $$ \tau(C) \equiv (....,\underline{IVT}_{j_{i-1}}^{p,k}(\overline{c_{i-1}}),\underline{IVT}_{j_i}^{p,k}(\overline{c_i}),\underline{IVT}_{j_{i+1}}^{p,k}(\overline{c_{i+1}}),....)$$
where $...j_{i-1}j_i,j_{i+1},... \in \{0,1,2,...,(p^{p^k}-1)\}$.
\begin{exm}
	Let the initial configuration of a CA having state set $\{0,1,2\}$ be $C_0=(01201)_3$ and the transition function  be given by $$\tau(01201)_3=\mu_1(101)\mu_2(012)\mu_3(120)\mu_4(201)\mu_5(010)=(02001)_3$$ where $\mu_1, \mu_3, \mu_5$ follow Wolfram code $377$ and  $\mu_2, \mu_4$ follow Wolfram code $588$ for $3$-state CA. \\
	A local transition function can be equivalently represented by some $\underline{IVT}_j^{3,3}$ and thus it follows that $\tau(01201)_3$ is equivalent to $$\underline{IVT}_{377}^{3,3}(1,0,1) \underline{IVT}_{588}^{3,3}(0,1,2) \underline{IVT}_{377}^{3,3}(1,2,0) \underline{IVT}_{588}^{3,3}(2,0,1) \underline{IVT}_{377}^{3,3}(0,1,0)$$
\end{exm}

Conversely, for any $\underline{IVT}_j^{p,k}$ if $k$ be odd, i.e. if $k=2r+1$ for some $r<k(\in \N_0)$, then $\underline{IVT}_{j}^{p,k}(\eta_1,\eta_2,...,\eta_{2r+1})$ will be equivalent to a local transition function of the $(r+1)^{th}$ cell in a CA with $p(\in \N)$ states and $2r+1$ neighbourhood whose underlying Wolfram code is  $j \in \{0,1,...,p^{p^k}-1\}$, given by  $$\underline{IVT}_{j}^{q,l}(\eta_1,\eta_2,...,\eta_{2r+1})=\underline{IVT}_{j}^{p,k}(\overline{\eta}_{r+1}) \equiv \mu_{r+1}(\overline{\eta}_{r+1})$$
\begin{rem}
	For a $p(\in \N)$-state $k(=2r+1 \in \N)$-neighbourhood CA, clearly $\forall \overline{c_i} \in Q^{S_i} \subseteq Q^{\Z}$ $\mu_{j}(\overline{c_i}) \equiv \underline{IVT}_{j}^{p,k}(\overline{c_i})$ where $j \in \{0,1,...,p^{p^k}-1\}$ is the underlying Wolfram Code. It follows that the set of local transition functions $M$ is equivalent to the set $T^{p,k}_{~|Q}$ and vice-versa.\\ Hence a function $\phi : M \rightarrow T^{p,k}_{~|Q}$ defined by $\phi(\mu_j)(\overline{c_i})=\underline{IVT}_{j}^{p,k}(\overline{c_i})$ is an isomorphism.\\
	Moreover, for any $v(\in \N \geq k)$-celled CA, having a global configuration $C=(c_1,c_2,...,c_v)$, if $\mathscr{Q}^{v}=\{\tau|\tau : Q^v \rightarrow Q^v\}$, then $\forall, \overline{c_i}\in Q^{k} \subseteq Q^{v},  i=1,2,...,v$, it follows that  a function $\phi :\mathscr{Q}^{v}  \rightarrow \mathscr{T}^{v}_{~|Q}$ defined by $$ \phi(\tau)(c_1,c_2,...,c_v) \equiv (\underline{IVT}_{j_1}^{p,k}(\overline{c_1}),\underline{IVT}_{j_2}^{p,k}(\overline{c_2}),...,\underline{IVT}_{j_v}^{p,k}(\overline{c_v}))$$ is an isomorphism. 
\end{rem}

\section{Some Algebraic Results on CA and IVT}

\begin{thm}
	$(\mathscr{Q}^{\Z}, \circ)$ forms a monoid w.r.t. composition of global transition functions.
\end{thm}
\begin{proof}
	Clearly $\mathscr{Q}^{\Z}$ is closed and associative under composition of global transition functions.\\
	The transition function $\mu_e$ such that $\forall c_i \in Q$, $\mu_{e}(\overline{c_{i}})=c_i$ is the local identity and it follows that $\tau_e \in \mathscr{Q}^\Z$ is the global identity such that $\forall C \in Q^\Z, i \in \Z$, $$\tau_e(C)= .....\mu_{e}(\overline{c_{i-1}})\mu_{e}(\overline{c_{i}})\mu_{e}(\overline{c_{i+1}}).....=C$$
	Hence the theorem. 
\end{proof}
\begin{cor}
	$(\overline{\mathscr{Q}^{\Z}}, \circ)$ forms a group w.r.t. composition of global transition functions w $\overline{\mathscr{Q}^{\Z}} \subseteq  \mathscr{Q}^{\Z}$ is the set of all invertible global transition functions of CA having state set $Q$.
\end{cor}
\begin{proof}
	Since any $\tau \in \overline{\mathscr{Q}^{\Z}}  \subseteq  \mathscr{Q}^{\Z}$ is invertible, the corollary holds true. 
\end{proof}

\begin{thm}
	$(\mathscr{T}^{v}, \circ)$ forms a monoid w.r.t. composition of IVTs.
\end{thm}
\begin{proof}
	The set $\mathscr{T}^{v}$, for $v \in \N$, is closed under composition $'\circ'$ since for any\\ $(IVT_{i_1}^{p,k},IVT_{i_2}^{p,k},...,IVT_{i_v}^{p,k}),(IVT_{j_1}^{p,k},IVT_{j_2}^{p,k},...,IVT_{j_v}^{p,k}) \in \mathscr{T}^{v}$ and for any\\ $(n_1,n_2...,n_v) \in \N_0^v$ $\exists~(m_1,m_2,...,m_v) \in \N_0^v$ such that,
	\begin{align*}
	\left((IVT_{i_1}^{p,k},IVT_{i_2}^{p,k},...,IVT_{i_v}^{p,k}) \circ (IVT_{j_1}^{p,k},IVT_{j_2}^{p,k},...,IVT_{j_v}^{p,k})\right)(n_1,n_2,...,n_v)\\=\left((IVT_{i_1}^{p,k}\circ IVT_{j_1}^{p,k}),(IVT_{i_2}^{p,k} \circ IVT_{j_2}^{p,k}),...,(IVT_{i_v}^{p,k} \circ IVT_{j_v}^{p,k})\right)(n_1,n_2,...n_v)\\=\left((IVT_{i_1}^{p,k}\circ IVT_{j_1}^{p,k})(\overline{n_1}),(IVT_{i_2}^{p,k} \circ IVT_{j_2}^{p,k})(\overline{n_2}),...,(IVT_{i_v}^{p,k} \circ IVT_{j_v}^{p,k})(\overline{n_v})\right)\\=\left(IVT_{i_1}^{p,k}(IVT_{j_1}^{p,k}(\overline{n_1})),IVT_{i_2}^{p,k}(IVT_{j_2}^{p,k}(\overline{n_2})),...,IVT_{i_v}^{p,k}(IVT_{j_v}^{p,k}(\overline{n_v}))\right)\\
	=IVT_{i_1}^{p,k}(\overline{n_1^*})IVT_{i_2}^{p,k}(\overline{n_2^*})...IVT_{i_v}^{p,k}(\overline{n_v^*})=(m_1,m_2,...,m_v)_p
	\end{align*}
	where $n_s^*=IVT_{j_s}^{p,k}(\overline{n_s})=IVT_{j_s}^{p,k}(\eta_1^{s},\eta_2^{s},...,\eta_k^{s})$ for $s=1,2,...,v$.\\
	$\mathscr{T}^{v}$ is associative under $'\circ'$ since composition of functions are associative.\\
	Again, for $(n_1,n_2...,n_v) \in \N_0^v$ $\exists ~(IVT_{e_1}^{p,k},IVT_{e_2}^{p,k},...,IVT_{e_v}^{p,k}) \in \mathscr{T}^{v}$ such that $$(IVT_{e_1}^{p,k},IVT_{e_2}^{p,k},...,IVT_{e_v}^{p,k})(n_1,n_2,...,n_v)$$ $$=IVT_{e_1}^{p,k}(\overline{n_1})IVT_{e_2}^{p,k}(\overline{n_2})...IVT_{e_v}^{p,k}(\overline{n_v})=(n_1,n_2,...,n_v)$$
	Since $IVT_{e_s}^{p,k}(\overline{n_s})=n_s~\forall s=1,...,v,$ it follows that $e_1=e_2=...=e_v=e(say)$\\
	Thus $(IVT_{e}^{p,k},IVT_{e}^{p,k},...,IVT_{e}^{p,k})$ is the identity element of $\mathscr{T}^{v}$\\
	Hence the theorem. 
\end{proof}

\begin{defn}
	Modular addition and multiplication of two local transition functions for a $p(\in \N)$-state CA are defined $\forall i \in \Z$ as $$(\mu^1_{i} \oplus_p \mu^2_{i})(\overline{c_{i}})=\mu^1_{i}(\overline{c_{i}}) \oplus_p \mu^2_{i}(\overline{c_{i}})~~and$$ $$(\mu^1_{i} \otimes_p \mu^2_{i})(\overline{c_{i}})=\mu^1_{i}(\overline{c_{i}}) \otimes_p \mu^2_{i}(\overline{c_{i}})$$ 
\end{defn}

\begin{defn}
	Modular addition and multiplication of two $p(\in \N)$-base $k(\in \N)$-dimensional IVTs(see \cite{hassan2012onedmensional}) are defined $\forall  (n_1,n_2,...,n_k) \in \N_0^k$ and $j_1, j_2 \in \{0,1,...,p^{p^k}-1\}$, as
	$$(IVT_{j_1}^{p,k} \oplus_p IVT_{j_2}^{p,k}) (n_1,n_2,...,n_k)=IVT_{j_1}^{p,k} (n_1,n_2,...,n_k) \oplus_p IVT_{j_2}^{p,k} (n_1,n_2,...,n_k)$$ $$=\left(f_{j_1}(a_0^{n_1},...,a_0^{n_k}) \oplus_p f_{j_2}(a_0^{n_1},...,a_0^{n_k})...f_{j_1}(a_{l-1}^{n_1},...,a_{l-1}^{n_k}) \oplus_p f_{j_2}(a_{l-1}^{n_1},...,a_{l-1}^{n_k})\right)_p~~and,$$
	$$(IVT_{j_1}^{p,k} \otimes_p IVT_{j_2}^{p,k}) (n_1,n_2,...,n_k)=IVT_{j_1}^{p,k} (n_1,n_2,...,n_k) \otimes_p IVT_{j_2}^{p,k} (n_1,n_2,...,n_k)$$ $$=\left(f_{j_1}(a_0^{n_1},...,a_0^{n_k}) \otimes_p f_{j_2}(a_0^{n_1},...,a_0^{n_k})...f_{j_1}(a_{l-1}^{n_1},...,a_{l-1}^{n_k}) \otimes_p f_{j_2}(a_{l-1}^{n_1},...,a_{l-1}^{n_k})\right)_p$$ where, $n_s=(a_0^{n_s},a_1^{n_s}...,a_{l-1}^{n_s})_p$ for $s=1,2,...,k$. 
\end{defn}

\begin{thm}
	$(\mathscr{Q}^{\Z}, \oplus_p, \otimes_p)$ forms a commutative ring with identity under the operations $\oplus_p$  and  $\otimes_p$ defined as $$(\tau_1 \oplus_p \tau_2)(C)=...(\mu^1_{i-1} \oplus_p \mu^2_{i-1})(\overline{c_{i-1}}).(\mu^1_{i} \oplus_p \mu^2_{i})(\overline{c_{i}}).(\mu^1_{i+1} \oplus_p \mu^2_{i+1})(\overline{c_{i+1}})...~and$$ $$(\tau_1 \otimes_p \tau_2)(C)=...(\mu^1_{i-1} \otimes_p \mu^2_{i-1})(\overline{c_{i-1}}).(\mu^1_{i} \otimes_p \mu^2_{i})(\overline{c_{i}}).(\mu^1_{i+1} \otimes_p \mu^2_{i+1})(\overline{c_{i+1}})...$$ where $\tau_1,\tau_2 \in \mathscr{Q}^{\Z} $, $C \in Q^\Z$,  $\oplus_p$ denotes addition modulo $p$ and $\otimes_p$ denotes multiplication modulo $p$  for $|Q|=p(\in \N) $.
\end{thm}
\begin{proof}
	In a $p(\in \N)$-state CA, for any $\tau_1, \tau_2 \in \mathscr{Q}^{\Z},$ $$(\tau_1 \oplus_p \tau_2)(C) \in Q^\Z~~and~~(\tau_1 \otimes_p \tau_2)(C) \in Q^\Z$$ since $\forall i \in \Z,$ $$(\mu^1_{i} \oplus_p \mu^2_{i})(\overline{c_{i}}) \in Q~~and~~(\mu^1_{i} \otimes_p \mu^2_{i})(\overline{c_{i}}) \in Q$$
	Therefore $\mathscr{Q}^{\Z}$ is closed w.r.t. $\oplus_p$ and $\otimes_p$.\\
	Associativity follows from associativity of $\oplus_p$ and $\otimes_p$.\\
	The local transition function $\mu_0$ such that $\forall c_i \in Q$, $\mu_{0}(\overline{c_{i}})=0$ is the additive identity since $\forall i \in \Z$ $$(\mu_{0} \oplus_p \mu_{i})(\overline{c_{i}})=(\mu_{i} \oplus_p \mu_{0})(\overline{c_{i}})=\mu_{i}(\overline{c_{i}})$$
	It follows that $\tau_0 \in \mathscr{Q}^\Z$ is the additive identity such that $\forall C \in Q^\Z$, $$\tau_0(C)= ...\mu_{0}(\overline{c_{i-1}})\mu_{0}(\overline{c_{i}})\mu_{0}(\overline{c_{i+1}})...$$
	The local transition function $\mu_{id}$ such that $\forall c_i \in Q$, $\mu_{id}(\overline{c_{i}})=1$ is the multiplicative identity since $\forall i \in \Z$ $$(\mu_{id} \otimes_p \mu_{i})(\overline{c_{i}})=(\mu_{i} \otimes_p \mu_{id})(\overline{c_{i}})=\mu_{i}(\overline{c_{i}})$$
	It follows that $\tau_{id} \in \mathscr{Q}^\Z$ is the multiplicative identity such that $\forall C \in Q^\Z$, $$\tau_{id}(C)= ...\mu_{id}(\overline{c_{i-1}})\mu_{id}(\overline{c_{i}})\mu_{id}(\overline{c_{i+1}})...$$
	
	Now any local transition function $\mu_i$, is a $p^k$-bit sequence $b_1b_2...b_{p^k}$(say), for a $p(\in \N)$-state $k(\in \N)$-neighbourhood CA where $b_1,...,b_{p^k} \in Q=\{0,1,...,p\}$. Clearly for any $b_s \in Q$, $\exists b_s^{-1} \in Q$ such that $\forall s=1,2,...,p^k$, $b_s \oplus_p b_s^{-1} =0$.\\ 
	Thus for any $\mu_i(\equiv b_1...b_{p^k})$, $\exists \mu_i^{-1}(\equiv b_1^{-1}...b_{p^k}^{-1})$ such that $\forall c_i \in Q$, $$(\mu_{i}^{-1} \oplus_p \mu_{i})(\overline{c_{i}})=(\mu_{i} \oplus_p \mu_{i}^{-1})(\overline{c_{i}})=\mu_{0}(\overline{c_{i}})=0$$ 
	It follows that for any $\tau \in \mathscr{Q}^\Z$, $\exists \tau^{-1} \in \mathscr{Q}^\Z$ such that $\forall C \in Q$, $$(\tau \oplus_p \tau^{-1})(C)=...(\mu_{i-1} \oplus_p \mu_{i-1}^{-1})(\overline{c_{i-1}})(\mu_{i} \oplus_p \mu_{i}^{-1})(\overline{c_{i}})(\mu_{i+1} \oplus_p \mu_{i+1}^{-1})(\overline{c_{i+1}})...=\tau_0(C),$$ $$(\tau^{-1} \oplus_p \tau)(C)=...(\mu_{i-1}^{-1} \oplus_p \mu_{i-1})(\overline{c_{i-1}})(\mu_{i}^{-1} \oplus_p \mu_{i})(\overline{c_{i}})(\mu_{i+1}^{-1} \oplus_p \mu_{i+1})(\overline{c_{i+1}})...=\tau_0(C),$$ 
	Thus for any global transition function $\tau$, its additive inverse exists in $\mathscr{Q}^\Z$.\\
	Commutativity follows from commutativity of $\oplus_p$ and $\otimes_p$.\\
	Now for any local transitions $\mu_i^1, \mu_i^2, \mu_i^3$, $\forall c_i \in Q, i \in \Z$ we get, $$(\mu_{i}^1 \otimes_p (\mu_{i}^2 \oplus_p \mu_{i}^3))(\overline{c_{i}})=(\mu_{i}^1 \otimes_p \mu_{i}^2)(\overline{c_{i}}) \oplus_p (\mu_{i}^1 \otimes_p \mu_{i}^3)(\overline{c_{i}})~~and,$$ $$((\mu_{i}^2 \oplus_p \mu_{i}^3)\otimes_p \mu_{i}^1)(\overline{c_{i}})=(\mu_{i}^2 \otimes_p \mu_{i}^1)(\overline{c_{i}}) \oplus_p (\mu_{i}^3 \otimes_p \mu_{i}^1)(\overline{c_{i}})$$ 
	Therefore for $\tau_1,\tau_2, \tau_3 \in \mathscr(Q)^\Z$, $\forall C \in Q^\Z$ we get, $$(\tau_1 \otimes_p (\tau_2 \oplus_p \tau_3))(C)=(\tau_1 \otimes_p \tau_2)(C) \oplus_p (\tau_1 \otimes_p \tau_3)(C)~~and,$$ $$((\tau_2 \oplus_p \tau_3)\otimes_p \tau_1)(C)=(\tau_2 \otimes_p \tau_1)(C) \oplus_p (\tau_3 \otimes_p \tau_1)(C)$$
	Hence the theorem.
\end{proof}
\begin{rem}The following example shows that $(\mathscr{Q}^{\Z}, \oplus_p, \otimes_p)$ will not form a field under the operations $\oplus_p$  and  $\otimes_p$ even if $|Q|=p$ is prime.
\end{rem}
\begin{exm}
	For initial configuration $C=(10001)_2$, let $$\tau(10001)_2=(\mu(110)\mu(100)\mu(000)\mu(001)\mu(011))_2$$
	Let $\mu$ follow Wolfram code $3(\neg (c_{i-1} \vee c_i))$.\\ Then $\tau(10001)_2=(\mu_3(110)\mu_3(100)\mu_3(000)\mu_3(001)\mu_3(011))_2=(00110)_2$.\\Now, $\tau_{id}(10001)_2=(\mu_{id}(110)\mu_{id}(100)\mu_{id}(000)\mu_{id}(001)\mu_{id}(011))_2=(11111)_2$.\\ But $\nexists$ $\tau^{-1}$ such that $(\tau \otimes_2 \tau^{-1})(10001)_2=(11111)_2$ since $\nexists$ $\mu^{-1}$ such that for all local configurations $\mu \otimes_2 \mu^{-1}=1$. 
\end{exm}

\begin{thm}
	$(\mathscr{T}^{v}, \oplus_p, \otimes_p)$ forms a commutative ring under the operations $\oplus_p$  and  $\otimes_p$ defined as
	$$\left((IVT_{i_1}^{p,k},IVT_{i_2}^{p,k},...,IVT_{i_v}^{p,k}) \oplus_p (IVT_{j_1}^{p,k},IVT_{j_2}^{p,k},...,IVT_{j_v}^{p,k})\right)(n_1,n_2,...,n_v)$$
	$$=\left((IVT_{i_1}^{p,k}\oplus_p IVT_{j_1}^{p,k}),(IVT_{i_2}^{p,k} \oplus_p IVT_{j_2}^{p,k}),.....,(IVT_{i_v}^{p,k} \oplus_p IVT_{j_v}^{p,k})\right)(n_1,n_2,...,n_v)$$
	$$=\left((IVT_{i_1}^{p,k}\oplus_p IVT_{j_1}^{p,k})(\overline{n_1}),(IVT_{i_2}^{p,k} \oplus_p IVT_{j_2}^{p,k})(\overline{n_2}),.....,(IVT_{i_v}^{p,k} \oplus_p IVT_{j_v}^{p,k})(\overline{n_v})\right)$$ $$and, \left((IVT_{i_1}^{p,k},IVT_{i_2}^{p,k},...,IVT_{i_v}^{p,k}) \otimes_p (IVT_{j_1}^{p,k},IVT_{j_2}^{p,k},...,IVT_{j_v}^{p,k})\right)(n_1,n_2,...,n_v)$$ 
	$$=\left((IVT_{i_1}^{p,k}\otimes_p IVT_{j_1}^{p,k}),(IVT_{i_2}^{p,k} \otimes_p IVT_{j_2}^{p,k}),....,(IVT_{i_v}^{p,k} \otimes_p IVT_{j_v}^{p,k})\right)(n_1,n_2,...,n_v)$$
	$$=\left((IVT_{i_1}^{p,k}\otimes_p IVT_{j_1}^{p,k})(\overline{n_1}),(IVT_{i_2}^{p,k} \otimes_p IVT_{j_2}^{p,k})(\overline{n_2}),....,(IVT_{i_v}^{p,k} \otimes_p IVT_{j_v}^{p,k})(\overline{n_v})\right)$$
	where $i_s,j_s \in \{0,1,...,p^{p^k}-1\},~p,k \in \N,~~\overline{n_s}=(\eta_1^{s},\eta_2^{s},...,\eta_k^{s})$ for $s=1,2,...,v$
\end{thm}
\begin{proof}
	For any $(IVT_{i_1}^{p,k},IVT_{i_2}^{p,k},...,IVT_{i_v}^{p,k}), (IVT_{j_1}^{p,k},IVT_{j_2}^{p,k},...,IVT_{j_v}^{p,k}) \in  \mathscr{T}^{v}$, and $(n_1,n_2,...,n_v) \in \N_0^v$, $$\left((IVT_{i_1}^{p,k},IVT_{i_2}^{p,k},...,IVT_{i_v}^{p,k}) \oplus_p (IVT_{j_1}^{p,k},IVT_{j_2}^{p,k},...,IVT_{j_v}^{p,k})\right)(n_1,n_2,...,n_v) \in \mathscr{\N_0}^{v}$$ $$and,\left((IVT_{i_1}^{p,k},IVT_{i_2}^{p,k},...,IVT_{i_v}^{p,k}) \otimes_p (IVT_{j_1}^{p,k},IVT_{j_2}^{p,k},...,IVT_{j_v}^{p,k})\right)(n_1,n_2,...,n_v)  \in \mathscr{\N_0}^{v}$$ since for every $s=1,2,...,v,$ $$(IVT_{i_s}^{p,k}\oplus_p IVT_{j_s}^{p,k})(\overline{n_s}) \in \N_0~~and~~(IVT_{i_s}^{p,k}\otimes_p IVT_{j_s}^{p,k})(\overline{n_s}) \in \N_0$$ where $\overline{n_s}=(\eta_1^{s},\eta_2^{s},...,\eta_k^{s})$\\ 
	Therefore $\mathscr{T}^{v}$ is closed w.r.t. $\oplus_p$ and $\otimes_p$ for $p \in \N$.\\
	Associativity follows from associativity of $\oplus_p$ and $\otimes_p$.\\
	$IVT_{0}^{p,k} \in T^{p,k}$ is such that $\forall \overline{n_s} \in \N_0^k,$~ $IVT_{0}^{p,k}(\overline{n_s})=0$  is the additive identity of $T^{p,k}$ since $\forall s=1,2,...,v$ $$(IVT_{0}^{p,k} \oplus_p IVT_{s}^{p,k})(\overline{n_s})= (IVT_{s}^{p,k} \oplus_p IVT_{0}^{p,k})(\overline{n_s})= IVT_{s}^{p,k}(\overline{n_s})$$
	where, for each $(\overline{n_s})= (\eta_1,\eta_2,...,\eta_k) \in \N_0^k$, $$IVT_{0}^{p,k} (\eta_1,\eta_2,...,\eta_k)=(f_0(a_0^{\eta_1},...,a_0^{\eta_k}),...,f_0(a_{l-1}^{\eta_1},...,a_{l-1}^{\eta_k}))=0$$  where, $\eta_u=(a_0^{\eta_u}a_1^{\eta_u}...a_{l-1}^{\eta_u})_p$ for $u=1,2,...,k.$\\
	It follows that $(IVT_{0}^{p,k},IVT_{0}^{p,k},...,IVT_{0}^{p,k})$  is the additive identity of  $\mathscr{T}^{v}$ such that for any   $(n_1,n_2,...,n_v) \in \N_0^v$ $$(IVT_{0}^{p,k},IVT_{0}^{p,k},...,IVT_{0}^{p,k})(n_1,n_2,...,n_v)=\left(IVT_{0}^{p,k}(\overline{n_1}),IVT_{0}^{p,k}(\overline{n_2}),...,IVT_{0}^{p,k}(\overline{n_v})\right)$$
	
	For any $IVT_{s}^{p,k} \in T^{p,k}$~ $\exists~IVT_{-s}^{p,k} \in T^{p,k}$ such that $\forall \overline{n_s} \in \N_0^k$, $$(IVT_{s}^{p,k} \oplus_p IVT_{-s}^{p,k})(\overline{n_s})= (IVT_{-s}^{p,k} \oplus_p IVT_{s}^{p,k})(\overline{n_s})= IVT_{0}^{p,k}(\overline{n_s})=0$$ 
	It follows that for any $(IVT_{i_1}^{p,k},IVT_{i_2}^{p,k},...,IVT_{i_v}^{p,k})\in \mathscr{T}^{v}$,~$\exists~(IVT_{-i_1}^{p,k},IVT_{-i_2}^{p,k},...,IVT_{-i_v}^{p,k}) \in \mathscr{T}^{v}$ such that $\forall (n_1,n_2,...,n_v) \in \N_0^v$, $$\left((IVT_{i_1}^{p,k},IVT_{i_2}^{p,k},...,IVT_{i_v}^{p,k}) \oplus_p (IVT_{-i_1}^{p,k},IVT_{-i_2}^{p,k},...,IVT_{-i_v}^{p,k})\right)(n_1,n_2,...,n_v)$$ $$=(IVT_{0}^{p,k},IVT_{0}^{p,k},...,IVT_{0}^{p,k})(n_1,n_2,...,n_v)$$ $$=\left((IVT_{-i_1}^{p,k},IVT_{-i_2}^{p,k},...,IVT_{-i_v}^{p,k}) \oplus_p (IVT_{i_1}^{p,k},IVT_{i_2}^{p,k},...,IVT_{i_v}^{p,k})\right)(n_1,n_2,...,n_v)$$
	Thus for any element of $\mathscr{T}^{v}$  its additive inverse exists in $\mathscr{T}^{v}$.\\
	Commutativity follows from commutativity of $\oplus_p$ and $\otimes_p$.\\
	
	Now for any $IVT_{h}^{p,k}, IVT_{i}^{p,k}, IVT_{j}^{p,k} \in T^{p,k}$, $\forall \overline{n_s} \in \N_0^k$, we get, $$(IVT_{h}^{p,k} \otimes_p (IVT_{i}^{p,k} \oplus_p IVT_{j}^{p,k}))(\overline{n_s})=(IVT_{h}^{p,k} \otimes_p IVT_{i}^{p,k})(\overline{n_s}) \oplus_p (IVT_{h}^{p,k} \oplus_p IVT_{j}^{p,k})(\overline{n_s}),$$
	$$((IVT_{i}^{p,k} \oplus_p IVT_{j}^{p,k}) \otimes_p IVT_{h}^{p,k})(\overline{n_s})=(IVT_{i}^{p,k} \otimes_p IVT_{h}^{p,k})(\overline{n_s}) \oplus_p (IVT_{j}^{p,k} \oplus_p IVT_{h}^{p,k})(\overline{n_s})$$
	Thus for any $(IVT_{h_1}^{p,k},...,IVT_{h_v}^{p,k}), (IVT_{i_1}^{p,k},...,IVT_{i_v}^{p,k}),(IVT_{j_1}^{p,k},...,IVT_{j_v}^{p,k})\in \mathscr{T}^{v},~\forall (n_1,...,n_v) \in \N_0^v$, 
	
	$$\left(IVT_{h_1}^{p,k},...,IVT_{h_v}^{p,k})\otimes_p \left((IVT_{i_1}^{p,k},...,IVT_{i_v}^{p,k})\oplus_p (IVT_{j_1}^{p,k},...,IVT_{j_v}^{p,k})\right)\right)(n_1,...,n_v)$$
	$$=\left((IVT_{h_1}^{p,k},...,IVT_{h_v}^{p,k}) \otimes_p (IVT_{i_1}^{p,k},...,IVT_{i_v}^{p,k})\right)(n_1,...,n_v)$$ $$\oplus_p \left((IVT_{h_1}^{p,k},...,IVT_{h_v}^{p,k}) \otimes_p (IVT_{j_1}^{p,k},...,IVT_{j_v}^{p,k})\right)(n_1,...,n_v), and$$ $$\left(\left((IVT_{i_1}^{p,k},...,IVT_{i_v}^{p,k})\oplus_p (IVT_{j_1}^{p,k},...,IVT_{j_p}^{p,k})\right) \otimes_p (IVT_{h_1}^{p,k},...,IVT_{h_v}^{p,k})\right)(n_1,...,n_v)$$ $$=\left((IVT_{i_1}^{p,k},...,IVT_{i_v}^{p,k}) \otimes_p (IVT_{h_1}^{p,k},...,IVT_{h_v}^{p,k})\right)(n_1,...,n_v)$$ $$\oplus_p \left((IVT_{j_1}^{p,k},...,IVT_{j_v}^{p,k}) \otimes_p (IVT_{h_1}^{p,k},...,IVT_{h_v}^{p,k})\right)(n_1,...,n_v)$$ 
	Hence the theorem.
\end{proof}

\begin{thm}
	$(\mathscr{T}^{v}_{~~|Q}, \oplus_p, \otimes_p)$ forms a commutative ring with identity under the operations $\oplus_p$  and  $\otimes_p$.
\end{thm}
\begin{proof}
	Clearly $(\mathscr{T}^{v}_{~~|Q}, \oplus_p, \otimes_p)$ for $p,k,v \in \N$ is a commutative ring since $\mathscr{T}^{v}_{~~|Q} \subseteq \mathscr{T}^{v}$ when $\N_0^v$ is restricted to $Q^v$.\\
	Now, for any $(\eta_1,\eta_2,...,\eta_k) \in Q^k$, $\exists$ multiplicative identity $\underline{IVT}_{id}^{p,k} \in T^{p,k}_{~|Q}$ such that $$\underline{IVT}_{id}^{p,k}(\eta_1,\eta_2,...,\eta_k)=\left(f_{id}(\eta_1,\eta_2,...,\eta_k)\right)_p= 1$$ $$since~(\underline{IVT}_j^{p,k}\otimes_p \underline{IVT}_{id}^{p,k})(\eta_1,\eta_2,...,\eta_k) = \underline{IVT}_j^{p,k}(\eta_1,\eta_2,...,\eta_k)= \left(f_{j}(\eta_1,\eta_2,...,\eta_k)\right)_p,$$ $$and~(\underline{IVT}_{id}^{p,k}\otimes_p \underline{IVT}_{j}^{p,k})(\eta_1,\eta_2,...,\eta_k)=\underline{IVT}_j^{p,k}(\eta_1,\eta_2,...,\eta_k)= \left(f_{j}(\eta_1,\eta_2,...,\eta_k)\right)_p $$
	It follows that $(\underline{IVT}_{id}^{p,k},\underline{IVT}_{id}^{p,k},...,\underline{IVT}_{id}^{p,k}) \in \mathscr{T}^{v}_{~~|Q}$ is the multiplicative identity of $\mathscr{T}^{v}_{~~|Q}$ such that $\forall (n_1,n_2,...,n_v) \in Q^v$, $$(\underline{IVT}_{id}^{p,k},\underline{IVT}_{id}^{p,k},...,\underline{IVT}_{id}^{p,k})(n_1,n_2,...,n_v)=\underline{IVT}_{id}^{p,k}(\overline{n_1})\underline{IVT}_{id}^{p,k}(\overline{n_2})...\underline{IVT}_{id}^{p,k}(\overline{n_v})$$ where $\underline{IVT}_{id}^{p,k}(\overline{n_s})=\underline{IVT}_{id}^{p,k}(\eta_1^s,\eta_2^s,...,\eta_k^s)=1$ for each $s=1,2,...,v$.\\
	Hence the theorem.
\end{proof}

\section{Modelling of DNA Sequence Evolution}
DNA can be modelled as a one-dimensional CA using rule matrix multiplication for linear CA(reported in \cite{mizas2008reconstruction}). The DNA sequence corresponds to the CA lattice and the deoxyribose sugars to the CA cells. The $4$ sugar bases A, C, T and G correspond to $4$ possible states of the CA cell.\\
Now, $1$-neighbourhood CA transition rule is practically absurd since the transition of cell states of a CA is dependent on neighbouring cells. Consequently, computation using $4$-state, $k(\in \N \geq 2)$-neighbourhood CA maybe complex. However, in the domain of IVTs, computations using one-dimensional IVTs hold good. Moreover,  for non-linear transition rules of a CA, the corresponding rule matrix maynot be obtained in general.\\ Hence, modelling DNA sequence evolution in terms of one-dimensional IVTs of $4$-base is a suitable alternative(see \cite{hassan2012dna}).\\  For any one-variable system having base $4$, there can be $4^{4^1}=256$ different functions of which $4^1=4$ are linear functions and the others non-linear. 
The following example depicts DNA sequence evolution in terms of IVTs where the sugar bases have been represented with numbers as follows: $$A\rightarrow 0, C\rightarrow 1, T\rightarrow 2, G\rightarrow 3$$
\begin{exm}
	Let CTCTAGAGGGAA be a particular DNA strand of length $12$ at some time. This strand corresponds to the number sequence $121203033300$.\\ Evolution of this strand at the next time step can be represented by any $(IVT_{j_1}^{4,1},...,IVT_{j_{12}}^{4,1}) \in \tau^{12}$ where $j_1,...j_{12} \in \{0,1,...,4^{4^1}-1\}$.\\Let the DNA strand be broken into $2$ blocks of $6$ sugar bases each.\\ If the first block evolves according to non-linear function $f_{78}(0\rightarrow 2,1\rightarrow 3,2\rightarrow 0,3\rightarrow 1)$ and the second block evolves according to linear function $f_{108}(0\rightarrow 0,1\rightarrow 3,2\rightarrow 2,3\rightarrow 1)$ then, $$(IVT_{j_1}^{4,1},...,IVT_{j_{12}}^{4,1})(\underbrace{121203}_{block 1}\underbrace{033300}_{block 2})_4$$ $$=\left(f_{78}(1)f_{78}(2)f_{78}(1)f_{78}(2)f_{78}(0)f_{78}(3)f_{108}(0)f_{108}(3)f_{108}(3)f_{108}(3)f_{108}(0)f_{108}(0)\right)_4$$ $$=(303021011100)_4$$ Thus the DNA strand evolved at the next step would be GAGATCACCCAA. Hence the evolution pattern at any later time step can be obtained recursively.   
\end{exm}

\section{Conclusion}
In this paper Wolfram code and transition function of a one-dimensional CA has been represented in terms of an IVT. Moreover, if the domain of IVTs is restricted from positive integers to $p(\in \N)$-base numbers, then an odd dimensional IVT is equivalent to a local transition function of a CA. It could be proved that the set of all transition functions of a $p(\in \N)$-state CA, forms a commutative ring with identity, whereas the set of $p$-base IVTs forms only a commutative ring under modular addition and multiplication. However, for the set of $p$-base IVTs multiplicative identity under modular addition and multiplication exists, when the domain of IVTs is restricted from positive integers to $p(\in \N)$-base numbers. The novelty of using IVTs instead of one-dimensional CA, for modelling DNA sequence evolution has been discussed and depicted in the last section. 

\bibliography{mybibfile}
\end{document}